\documentclass[12pt,letterpaper]{article}

\usepackage[includeheadfoot,
            marginratio={1:1,2:3}, 
            width=412pt, 
            height=688pt,]{geometry}

\usepackage{amsmath}
\usepackage{amsfonts}
\usepackage{amssymb}
\usepackage{graphicx}
\usepackage{cite}
\usepackage{ulem}
\usepackage{longtable}
\usepackage{afterpage}
\usepackage{amsthm}

\newcommand{\eq}[1]{\begin{equation}
                     \begin{aligned} #1 \end{aligned}
                     \end{equation}}

\allowdisplaybreaks[2]
\numberwithin{equation}{section}

\theoremstyle{definition}
\newtheorem{defn}{Definition}
\theoremstyle{plain}
\newtheorem{theo}{Theorem}
\theoremstyle{plain}
\newtheorem{lem}{Lemma}
\theoremstyle{remark}

\theoremstyle{remark}

\theoremstyle{remark}

\theoremstyle{plain}

\begin{document}

\normalem
\vspace*{-1.5cm}
\begin{flushright}
  {\small
  MPP-2018-146 \\
  }
\end{flushright}

\vspace{1.5cm}

\begin{center}
  {\LARGE
   On the Uniqueness of L$_\infty$ bootstrap:\\[0.3cm] 
   Quasi-isomorphisms are Seiberg-Witten Maps
}
\vspace{0.4cm}

\end{center}

\vspace{0.35cm}
\begin{center}
  Ralph Blumenhagen$^{1}$, Max Brinkmann$^{1}$,\\[0.1cm]  Vladislav Kupriyanov$^{1,2,3}$, Matthias Traube$^{1}$
\end{center}

\vspace{0.1cm}
\begin{center} 
\emph{$^{1}$ Max-Planck-Institut f\"ur Physik (Werner-Heisenberg-Institut), \\ 
   F\"ohringer Ring 6,  80805 M\"unchen, Germany } \\[0.1cm] 
\vspace{0.25cm} 
\emph{$^{2}$ CMCC-Universidade Federal do ABC, Santo Andr\'e, SP, 
Brazil}\\[0.1cm] 
\vspace{0.25cm} 
\emph{$^{3}$ Tomsk State University, Tomsk, Russia}\\

\vspace{0.0cm}

 \vspace{0.3cm} 
\end{center}

\vspace{1cm}

\begin{abstract}
\noindent
In the context of the recently proposed L$_\infty$ bootstrap approach, the
question arises whether the so constructed gauge theories are unique
solutions of the L$_\infty$ relations. Physically it is expected that 
two gauge theories should be considered equivalent if they are related
by a field redefinition described  by a Seiberg-Witten map. To clarify the 
consequences in the  L$_\infty$ framework, 
it is proven that Seiberg-Witten maps between physically equivalent
gauge theories  correspond to  quasi-isomorphisms of the underlying L$_\infty$
algebras.  The proof suggests  an extension of the definition of a Seiberg-Witten
map to the closure conditions of two gauge transformations 
and the dynamical equations of motion. 
\end{abstract}

\clearpage


\section{Introduction}

The connection between gauge theories and L$_\infty$ algebras goes
back to Zwiebach's formulation of Closed String field theory (CSFT)
\cite{Zwiebach:1992ie}. The idea was to consistently decompose the
moduli space of compact Riemann surfaces without boundaries into closed
string vertices and propagators. There each vertex corresponds to
a function taking $n$ fields as input and giving back one output
field. Using these functions, an action for bosonic CSFT was written
down. In order to quantize the theory the action should satisfy the
quantum BV-master action which, roughly speaking,  is the case if the
string $n$-products of the action define a loop L$_\infty$ algebra
\cite{Markl:1997bj}. 

When only looking at tree level diagrams the
action should satisfy the classical master equation which yields an
L$_\infty$ algebra for the $n$-products. Using the geometric language of
formal supermanifolds, in \cite{Alexandrov:1995kv} the connection
between the classical master action of the BV-formalism and L$_\infty$
algebras was further analyzed, leading to the physics folklore that
any consistent classical gauge theory should define an L$_\infty$
algebra. This picture was further elucidated in the more recent work
by Hohm and Zwiebach \cite{Hohm:2017pnh}, where it was explicitly
shown how  both the gauge transformation rules and the equations of 
motion of Chern-Simons and Yang-Mills theories are encoded in 
underlying  L$_\infty$ algebras.
New connections between
L$_\infty$ algebras and extended symmetry algebras, called  W-algebras, of
two-dimensional  conformal field theories were established
\cite{Blumenhagen:2017ulg,Blumenhagen:2017ogh}. 

Besides the
$n$-product defining a vertex in the action, to actually write down an action one also needs  an
appropriate inner product. 
When trying to construct gauge theories in situations where the notion
of an inner product and hence an action is a priori not clear it is
still a good guiding principle to demand that the constructed theory
should define an L$_\infty$ algebra. This point of view was applied in
a recent paper \cite{Blumenhagen:2018kwq} where the authors propose an
L$_\infty$ bootstrap program for the construction of non-commutative
gauge theories with non-constant non-commutativity parameter. 
Starting with a free theory and a fixed gauge group
the bootstrap idea proceeds by first making an ansatz for the lowest order
products and  then by solving the L$_\infty$ relations order by
order. This in turn defines the higher order L$_\infty$ products.
By this procedure one can derive both  
interaction terms in the resulting equations of motion and higher 
products describing the gauge transformations of the fields. 

A natural
question is if this procedure leads to a unique solution and how
physically equivalent solutions are related. From one side, two L$_\infty$ algebras are known to be equivalent if they are quasi-isomorphic~\cite{Lada:1994mn}. From the other, following
\cite{Seiberg:1999vs},  gauge theories describing the same physics
should be related via a Seiberg-Witten map, which ensures that gauge
orbits of the respective theories get mapped onto  each other. Seiberg-Witten maps 
have been studied in terms of the antibracket formalism for gauge theories (see \cite{Barnich:2003wq} and references therein), which is related to L$_\infty$ algebras. In this note we show that a Seiberg-Witten map is  the same as having a specific subset of quasi-isomorphism (QISO) relations of the L$_\infty$ algebras underlying the two
gauge theories.

We first motivate this connection with a simple
example from an abelian Chern-Simons theory. Then we  formulate and prove
a general theorem stating clearly the equivalence of the physical notion
of a Seiberg-Witten map and the mathematical notion of an  L$_\infty$
QISO. 
This means in particular that
different solutions to the bootstrap program yield equivalent physics
if they are related via an L$_\infty$ quasi-isomorphism. Let us
emphasize that this does not show uniqueness of the
bootstrap\footnote{In other words,  L$_\infty$ QISOs define equivalence classes of L$_\infty$ algebras. To illustrate this we remind the reader of a similar situation with the definition of the star product. Two star products $\star$ and $\star'$ are equivalent if they are related by a ``gauge'' transformation
$
f\star' g= \mathcal{D}^{-1}\big( \mathcal{D}f\star \mathcal{D}g
\big)$, where $\mathcal{D}$ is a QISO \cite{Kontsevich:1997vb}. In
general, $\star$ and $\star'$ may have very different expressions and
properties, e.g., $\star$ can be closed with respect to the integral
$\int f\star g=\int f\cdot g$, while $\star'$ is not. However some
features are conserved, if $\star$ is associative and represents the
quantization of the Poisson bracket $\{f,g\}$, then $\star'$ is also
associative and represents a different quantization of the same
bracket. Thus star products related by QISOs represent different
quantization prescriptions of the same classical system. This is
similar to our statement that QISOs relate physically equivalent gauge
theories via SW maps.}.  

In \cite{Barnich:1993vg} (see \cite{Barnich:2000zw} for a complete
review) the question of consistent interactions for classical
Yang-Mills gauge theory was addressed using local BRST-cohomology. This corresponds to an L$_\infty$ bootstrap starting with
a free Yang-Mills theory. The results there depend crucially on the 
concrete  gauge theory under consideration, hence it seems unlikely
that  there exist a general answer to the question of uniqueness. 
Instead, in this paper we take a more moderate step and investigate
solely the question  when  two a priori different solutions to the
bootstrap describe the  same physics.

The paper is organized as follows. In section 2 we briefly recall the
definition of L$_\infty$ quasi-isomorphisms and Seiberg-Witten
maps. In section 3, we recall the bootstrap approach and present a
simple example   showing
how possible redundancies can appear when solving the L$_\infty$
equations. We point out that these are related to Seiberg-Witten maps 
 and to the existence of  quasi-isomorphisms. Building on these examples we give a general
proof of the equivalence of the two structures in section
\ref{sec_main}. We realize that from the QISO structure, one can learn
some new aspects of the Seiberg-Witten map, namely how the closure
conditions of the two related gauge theories are mapped and how 
the equations of motion are related.

\section{L$_\infty$-QISO and Seiberg-Witten maps}

In this section we set the stage and present the relevant definitions
and properties of  L$_\infty$ algebras, their quasi-isomorphisms and Seiberg-Witten
maps. We will start with some aspects of  L$_\infty$ algebras.

\subsection{Basics of L$_\infty$ algebras}
\label{sec_basicL}

For convenience, as signs are simpler, we will work
with the following definition \cite{Kajiura:2004xu} of an L$_\infty$ algebra in the so-called
$b$-picture \cite{Hohm:2017pnh}:
\begin{defn}
\label{defiLinfty}
An \textbf{L$_\infty$ algebra in the b-picture} consists of a graded vector space $V=\underset{i\in \mathbb{Z}}{\bigoplus}\, V_i$ equipped with graded symmetric multilinear maps $b_n:V^{\otimes\, n}\rightarrow V$ of constant degree $|b_n|=-1$ satisfying 
\eq{
\label{jacobigen}
\underset{\sigma\in \, {\rm Unsh}(k+l=n)}{\sum} \, \epsilon(\sigma;x)\, b_{1+l}\Big(b_k(x_{\sigma(1)},\dots, x_{\sigma(k)}),x_{\sigma(k+1)},\dots,x_{\sigma(n)}\Big)=0
}
with $\epsilon(\sigma;x)$ denoting the Koszul sign defined via
\eq{
x_{\sigma(1)}\wedge\dots\wedge x_{\sigma(n)}=\epsilon(\sigma;x)\, x_1\wedge\dots\wedge x_n\; .
}
Here we used $x_i\wedge x_j=(-1)^{x_i x_j}\, x_j\wedge x_i$ where an $x_i$ in the
exponent  stands for the degree of the corresponding
element. Moreover, the sum in \eqref{jacobigen} runs over the
unshuffled permutations that satisfy
\eq{
           \sigma(1)<\ldots <\sigma(k)\,,\qquad \sigma(k+1)<\ldots
           <\sigma(n) \,.
}
\end{defn}

\vspace{0.3cm}
In contrast to the definition in \cite{Kajiura:2004xu} we used a
homological grading on the graded vector space which merely results in
having the maps being of degree $-1$. The first few defining equations
explicitly  read
\begin{itemize}
\item[$n$=1:]\label{firstequ}
\eq{
0=b_1\left(b_1(x)\right)
}
\item[$n$=2:]
\eq{\label{secondequ}
0=b_1\left(b_2(x_1,x_2)\right)+b_2\left(b_1(x_1),x_2\right)+(-1)^{x_1}\, b_2\left(x_1,b_1(x_2)\right)
}
\item[$n$=3:]\label{thirdequ}
\eq{
0&=b_1\left(b_3(x_1,x_2,x_3)\right)+b_3\left(b_1(x_1),x_2,x_3\right)+(-1)^{x_1}b_3\left(x_1,b_1(x_2),x_3\right)\\
&\phantom{}+(-1)^{x_1+x_2}\,
b_3\left(x_1,x_2,b_1(x_3)\right)+b_2\left(b_2(x_1,x_2),x_3\right)\\
&\phantom{}+(-1)^{x_3(x_2+x_1)}\, b_2\left(b_2(x_3,x_1),x_2\right)+(-1)^{x_1(x_2+x_3)}\, b_2\left(b_2(x_2,x_3),x_1\right).
}
\end{itemize}
Like in the usual definition ($\ell$-picture), the first
map is a differential. The second equation is like a Leibniz rule
between the differential and the two-bracket, but now with slightly
unusual signs. The same applies for the third equation, which still
means that the two-bracket satisfies the Jacobi identity up to
homotopy. The crucial difference to the $\ell$-picture is that the $b$-maps are graded
symmetric and of constant degree $-1$. 
Having this definition at hand we can now follow \cite{Kajiura:2004xu} 
and provide the notion of an L$_\infty$ morphism\footnote{In \cite{Lada:1994mn} it was shown that the $\ell$-picture of an
 L$_\infty$ algebra is equivalent to having a nilpotent coderivation $Q$ of degree
$-1$ on the coalgebra over the suspension of the graded vector
space. A morphism between two L$_\infty$ algebras $(C(V),Q_V)$ and
$(C(W),Q_W)$ is then a cohomomorphism $F:C(V)\rightarrow C(W)$ of the
coalgebras satisfying $Q_W\circ F=F\circ Q_V$. Although this gives a
closed expression of an L$_\infty$ morphism this expression is not
very useful for actual calculations.}:

\begin{defn}
Let $(V, \left\lbrace b_i \right\rbrace)$ and $(W,\big\lbrace
  \tilde{b}_j \big\rbrace)$ be L$_\infty$ algebras. 
\vspace{-0.1cm}
\begin{itemize}
\item[1)] An \textbf{L$_\infty$ morphism} $F:(V, \left\lbrace b_i
  \right\rbrace)\rightarrow (W,\big\lbrace \tilde{b}_i \big\rbrace)$
  consists of multilinear, graded symmetric maps $\left\lbrace F_n
  \right\rbrace: \bigotimes^n V\rightarrow W$ of degree $|F_n|=0$ 
such that
\eq{\label{defiquasi}
&\underset{\sigma \in\, {\rm Unsh}(k+l=n)}{\sum} \, \epsilon(\sigma; x)\, F_{1+l}\Big(b_k(x_{\sigma(1)},\dots, x_{\sigma(k)}),x_{\sigma(k+1)},\dots, x_{\sigma(n)}\Big)\\
=&\underset{\sigma \in \, {\rm Unsh}(k_1+\dots+k_j =n)}{\sum}\; \frac{\epsilon(\sigma;x)}{j!}\, \tilde{b}_j(F_{k_1}\otimes\dots\otimes F_{k_j})(x_{\sigma(K)})\; .
}
In the second line the index of the entry is a multindex $K=(k_1,\dots,
k_j)$ of length $n$. It keeps track of the partition of the $n$ entries
$(x_1,\dots, x_n)$ inserted into the $j$ maps $F_{k_i}$. 
The unshuffles $\sigma(K)$ consist of all permutations keeping the entries in the $j$-partitions ordered.  

\item[2)]  An \textbf{L$_\infty$ quasi-isomorphism} is an L$_\infty$
  morphism $F$, whose lowest map $F_1:~V\rightarrow W$ induces an
  isomorphism on the homology of the chain complexes underlying the
  L$_\infty$ algebras.
\end{itemize}

\end{defn}

\vspace{0.3cm}
\noindent
Concretely, at lowest order the definition of an L$_\infty$ morphism gives the
following equations 
\begin{itemize}
\item[$n$=1:] 
\eq{\label{QISO1}
F_1(b_1(x))=\tilde{b}_1(F_1(x))
}
\item[$n$=2:]
\eq{\label{QISO2GEN}
&F_1(b_2(x_1,x_2))+F_2(b_1(x_1),x_2)+(-1)^{x_1} F_2(x_1,b_1(x_2))\\
=&\,
\tilde{b}_1(F_2(x_1,x_2))+{\textstyle\frac{1}{2}}\tilde{b}_2(F_1(x_1),F_1(x_2))+{\textstyle
  \frac{1}{2}}\tilde{b}_2(F_1(x_1),F_1(x_2))\,.
}
\end{itemize}

\noindent
Since later we will choose $F_1={\rm id}$, (i.e. we are considering
automorphisms of L$_\infty$ algebras on a given chain complex), the
second equation can be rewritten as
\eq{\label{QISO2}
\tilde{b}_2(x_1,x_2)-&b_2(x_1,x_2)=\\
&F_2(b_1(x_1),x_2)+(-1)^{x_1}F_2(x_1,b_1(x_2))-\tilde b_1(F_2(x_1,x_2))\,.
}
Thus for two L$_\infty$ algebras $(V, \left\lbrace  b_i \right\rbrace )$ and 
$(V,\big\lbrace \tilde{b}_i\big\rbrace )$ with $b_1=\tilde{b}_1$ to be
quasi-isomorphic the product $\tilde{b}_2$ must be decomposable in the
form \eqref{QISO2}. It is not too hard to see, that this procedure
also works in higher degree equations. The right hand side of
\eqref{defiquasi} always contains terms with $j=n$, i.e. a summand  of the form
\eq{ 
\underset{\sigma \in S_n}{\sum} \frac{\epsilon(\sigma;x)}{n!}\, \tilde{b}_n(F_1\otimes\dots\otimes F_1)(x_{\sigma(1)},\dots, x_{\sigma(n)})\; .
}
The sum runs over all permutations of $\left\lbrace 1,\dots, n
\right\rbrace$. As in the step from \eqref{QISO2GEN} to \eqref{QISO2}
we use $F_1={\rm id}$. Upon interchanging the entries, all terms are equal
in this case and a single term $\tilde{b}_n(x_1,\dots,x_n)$ is left.
This shows  that the defining relation can always be solved for the
highest appearing product $\tilde{b}_n$. Bringing all other terms to
the left produces a necessary and sufficient condition on the products
$\tilde{b}_n$ for being  related to the $b$'s via a quasi-isomorphism
\eq{\label{PRNEWQISO}
&\tilde{b}_n(x_1,\dots, x_n)-b_n(x_1,\dots,x_n)\\
=&\underset{\sigma \in\, \underset{k<n}{{\rm Unsh}(k+l=n)}}{\sum} \,
\epsilon(\sigma; x)\, F_{1+l}\Big(b_k(x_{\sigma(1)},\dots,
x_{\sigma(k)}),x_{\sigma(k+1)},\dots,
x_{\sigma(n)}\Big)\bigg|_{F_1={\rm id}}\\
&-\underset{\sigma \in \, {\rm Unsh}(k_1+\dots +k_j =n), j<n}{\sum}\;
\frac{\epsilon(\sigma;x)}{j!}\, \tilde{b}_j(F_{k_1}\otimes\dots\otimes
F_{k_j})(x_{\sigma(K)})\bigg|_{F_1={\rm id}}\; .
}
Since in the bootstrap approach one solves the L$_\infty$ equations in
an iterative way starting with the lowest products, the above relation
can be applied for  relating  two possible solutions $b$ and $\tilde{b}$.
We will provide a simple example for this procedure in section \ref{sec_example}.

\subsection{Seiberg-Witten maps}

Next we recall the  notion of a Seiberg-Witten (SW) map between gauge
theories. In their seminal paper \cite{Seiberg:1999vs} Seiberg and
Witten analyzed the behavior of open strings in backgrounds with
non-zero 
but constant Kalb-Ramond $B$-field. It turns out that the $B$-dependence can be
completely captured by making space-time non-commutative, i.e. by
introducing the Moyal-Weyl star product between functions with
$\theta\sim B^{-1}$. It is then further argued that in the limit
of large $B$-field one arrives at a description in terms of non-commutative
Yang-Mills (YM) theory with the  Moyal-Weyl star product. As
this corresponds to a specific choice of regularization of the world
sheet theory, Seiberg and Witten argued that there has to be a one to
one correspondence between non-commutative and  ordinary Yang-Mills theory. 

Recall that non-commutative Yang Mills theory on the Moyal-Weyl plane
with a star product 
\eq{
(f\star g)(x):= e^{\frac{i}{2}\theta^{ij}\frac{\partial}{\partial y^i}\frac{\partial}{\partial z_j}}\,  f(y)g(z)\, \vert_{y=z=x}
}
formally looks like usual YM theory where all point-wise
products are replaced by star products. This means that for a
non-commutative gauge field $\hat{A}$ the infinitesimal gauge
variation takes the form
\eq{\label{noncommutrafoa}
\hat{\delta}_{\hat{\lambda}}\hat{A}_j=\partial_j \hat{\lambda}+i\,
\hat{\lambda}\star \hat{A}_j- i\,\hat{ A}_j\star\hat{\lambda}
}
while the field strength is defined as
\eq{
\label{noncommutrafob}
\hat{F}_{jk}=\partial_j\hat{A}_k-\partial_{k}\hat{A}_j-i\, \hat{A}_j\star\hat{A}_k+i\, \hat{A}_k\star\hat{A}_j\,.
}
As pointed out in  \cite{Seiberg:1999vs}, in order to relate non-commutative Yang-Mills theory to ordinary Yang-Mills
it suffices that the gauge orbits of the respective theories
are  mapped onto each other. This ensures that the degrees of
freedom on both sides are the same\footnote{Note that this does not imply that
the gauge groups are equivalent, as can  be inferred
by looking at a  non-commutative $U(1)$ YM theory,  which is
non-abelian.}. The relation between the non-commutative and
ordinary YM theory turned out to be of the following type.

\begin{defn}
\label{defSWmap}
Two gauge theories with data $(\lambda,A)$ and
$(\hat\lambda,\hat A)$  are related via a  Seiberg-Witten map, if
there exist two maps 
\eq{
   \hat{\lambda}=\hat{\lambda}(\lambda,A)\,,\qquad
    \hat{A}=\hat{A}(A)
}
so that (at linear order in $\lambda$) gauge orbits are mapped onto gauge orbits
\eq{\label{SWMAP1}
\hat{A}(A+\delta_\lambda A)=
\hat{A}(A)+\hat{\delta}_{\hat{\lambda}(\lambda, A)}\hat{A}(A)\; .
}
\end{defn}

\vspace{0.3cm}
\noindent
For the concrete case of non-commutative YM theory, this map has been worked
out order by order in $\theta^{ij}$~\cite{Seiberg:1999vs}. As we
will see later when comparing L$_\infty$ QISO to SW maps, 
this definition can be extended by further requirements.
To see how, let us consider the closure of two  gauge transformations
\eq{
           \left[ \delta_{\lambda_1}, \delta_{\lambda_2} \right]  A =\delta_{[\lambda_1,\lambda_2]} A\,,\qquad 
            [\delta_{\hat\lambda_1},\delta_{\hat\lambda_2}] \hat A =\delta_{[\hat\lambda_1,\hat\lambda_2]_\star} \hat A
}
on both sides of the  the original SW duality. 
The naive guess for the mapping between the two closure conditions would be 
\eq{
\hat{A}\left(A+\delta_{ [\lambda_1,\lambda_2] }A\right)=\hat{A}(A)+\hat{\delta}_{[\hat{\lambda}_1,\hat{\lambda}_2]_{\star}}\hat{A}(A)\,.
}
However,  by inspection of the $U(1)$ case this cannot  be
true. The left-hand side  of the equation is 
$\hat{A}(A)$ as $U(1)$ is
abelian, which is not true for the non-commutative $U(1)$ on the right
hand side. Hence the
second term on the right side is non-vanishing and the equation cannot
hold. However, there is something we can say about the closure
relation from \eqref{SWMAP1}
\eq{
\label{SWCLOSURE1}
\hat{A}\left(A+\delta_{ [\lambda_1,\lambda_2] }A\right)=\hat{A}(A)+\hat{\delta}_{\hat{\lambda}([\lambda_1,\lambda_2],A)}\hat{A}(A)\; .
}
Thus the question is what $\hat{\lambda}([\lambda_1,\lambda_2],A)$
really is.  We can use the explicit form of the original
SW map up to first order in the non-commutativity parameter
$\theta^{ij}$,(c.f. \cite{Seiberg:1999vs} (3.5)) to derive an equation to
first order. An elementary but tedious computation shows that the
following relation holds 
\eq{
\label{commuexp}
\hat{\lambda}([\lambda_1,\lambda_2],A)\Big\vert_{\mathcal{O}(\theta)}=\left. \big[\hat{\lambda}_1,\hat{\lambda}_2\big]_{\star}\right\vert_{{\mathcal{O}(\theta)}}+\left. \hat{\lambda}(\lambda_1,\delta_{\lambda_2}A)-\hat{\lambda}(\lambda_2,\delta_{\lambda_1}A)\right\vert_{\mathcal{O}(\theta)}\; .
}
Holding up to first order in $\theta^{ij}$, it is a reasonable extension to
require \eqref{commuexp} to hold to all orders. 
Therefore, we conjecture that the gauge closures should map as 
\eq{\label{SWCLOSURE2}
\hat{A}\left(A+\delta_{ [\lambda_1,\lambda_2]
  }A\right)=\hat{A}(A)&+\hat{\delta}_{[\hat{\lambda}_1,\hat{\lambda}_2]_{\star}}\hat{A}(A)\\[0.1cm]
&+\hat{\delta}_{\hat{\lambda}(\lambda_1,\delta_{\lambda_2}A)}\hat{A}(A)-\hat{\delta}_{\hat{\lambda}(\lambda_2,\delta_{\lambda_1}A)}\hat{A}(A)\; .
}
In the main section \ref{sec_main}, by exploiting the intriguing relation of the
SW map to  an L$_\infty$  QISO we collect further evidence that this
is the right formula. For gauge transformations that  close 
only on-shell there will be an extra term in \eqref{SWCLOSURE2}.
Moreover, we will find how SW maps should be 
extended to the equations of motion of the two equivalent theories.

\section{Redundancies in the L$_\infty$ bootstrap}

We begin with recalling how a classical gauge theory with irreducible
gauge freedom can be described in terms of an L$_\infty$ algebra. By
irreducible gauge freedom we mean that the gauge parameters of the
field theory do not have gauge redundancies themselves. This very much
builds on the dictionary established in \cite{Hohm:2017pnh}. 

\subsection{Basics of the L$_\infty$ bootstrap approach}

In order
to define an L$_\infty$ algebra in the sense of definition
\eqref{defiLinfty} we need a graded vector space. We take
$X=X_{1}\oplus X_0\oplus X_{-1}$ and all others trivial. The
assignment is as follows. The vector space $X_1$ contains the gauge
parameters, $X_0$ is the space of fields and $X_{-1}$ contains the
equations of motion of the gauge theory. A standard gauge
transformation is then of the form
\eq{\label{gaugetrafodefi}
\delta_\lambda A=b_1(\lambda)+b_2(\lambda,A)+\frac{1}{2}b_3(\lambda,A,A)+\dots =\underset{n=0}{\overset{\infty}{\sum}}\, \frac{1}{n!}b_{n+1}(\lambda,A^n)\, . 
}
The equations of motion can be written as 
\eq{
\mathcal{F}=b_1(A)+\frac{1}{2}b_2(A,A)+\frac{1}{3!}b_3(A,A,A)+\dots=\underset{n=1}{\overset{\infty}{\sum}} \frac{1}{n!}\, b_n(A^n)\; .
}
Then, the equation of motion transforms under gauge transformations as follows
\eq{
\delta_\lambda
\mathcal{F}=b_2(\lambda,\mathcal{F})+b_3(\lambda,\mathcal{F},A)
+\dots =\sum_{n=0}^\infty\, \frac{1}{n!}\, b_{n+2}(\lambda,\mathcal{F},A^n)\; .
}
Using the L$_\infty$ equations one can show that the gauge commutator yields \cite{Hohm:2017pnh,Fulp:2002kk}
\eq{\label{closure}
\left[ \delta_{\lambda_1}, \delta_{\lambda_2} \right] A\sim
\delta_{C(\lambda_1,\lambda_2,A) }\, A+ \delta^T_{C(\lambda_1,\lambda_2,{\cal F},A)}\,A
}
with
\eq{
  C(\lambda_1,\lambda_2,A)=
  b_2(\lambda_1,\lambda_2)+b_3(\lambda_1,\lambda_2,A)+\ldots=
   \sum_{n=0}^\infty\, \frac{1}{n!}\, b_{n+2}(\lambda_1,\lambda_2,A^n)
}
and where the second term on the right hand side of \eqref{closure}
vanishes on-shell. It can be expanded as
\eq{
\label{defieomtrafo}
  C(\lambda_1,\lambda_2,{\cal F},A)=
  b_3(\lambda_1,\lambda_2,{\cal F})
+\ldots=
   \sum_{n=0}^\infty\, \frac{1}{n!}\, b_{n+3}(\lambda_1,\lambda_2,{\cal F},A^n)\,.
}
We come back to this term in section \ref{sec_main}, but for now we
drop it and assume that the gauge structure defines a closed
algebra. Comparing \eqref{gaugetrafodefi} to \eqref{noncommutrafoa},
for the non-commutative YM theory one can read off
\eq{
b_1(\hat\lambda)=\partial \hat\lambda, \qquad b_2(\hat\lambda,\hat
A)=i \, \hat{\lambda}\star \hat{A}_j- i\,\hat{ A}_j\star\hat{\lambda},
\qquad b_n(\hat\lambda,\hat A^n)=0,\, n\ge 2\; .
}

Instead of reading-off the L$_\infty$ maps from a given gauge
structure and equations of motion and checking the L$_\infty$
equations one could proceed in the other direction. This the 
L$_\infty$ bootstrap approach followed in \cite{Blumenhagen:2018kwq}. Let us explain
briefly how this works, for more details we refer to \cite{Blumenhagen:2018kwq}.

One initially  starts with a free theory, i.e. an assignment of $b_1(\lambda), b_1(A)$ and a
fixed gauge group represented through its algebra captured in
$b_2(\lambda_1,\lambda_2)$ and the other products left open a priori. Trying to solve
subsequently the L$_\infty$ equation will fix higher products and
define a consistent theory by the arguments above. Due to the grading of the elements a specific
L$_\infty$ equation has to be solved only on particular inputs. Note
that every L$_\infty$ relation in the $b$-picture will lower the degree
by two. Take for example  the equation with two inputs \eqref{secondequ}. It
is only non-trivial on insertions $\lambda_1,\lambda_2\in X_{1}$ of total degree
$2$ and $\lambda\in X_{1}, A\in X_{0}$ of total degree $1$. The first case yields
\eq{
0=b_1(b_2(\lambda_1,\lambda_2))+b_2(b_1(\lambda_1),\lambda_2)-b_2(\lambda_1,b_1(\lambda_2))\; .
}
The first term is fixed by our assignments, but we are free to chose
$b_2(\lambda,A)$ such that the above holds with $A=b_1(\lambda)$. Hence we
are defining a new product corresponding to the transformation law
of the field. In the same fashion the insertion of $\lambda,A$ will then fix
a product $b_2(A,B)$ with $A,B \in X_0$. 

The question we are concerned with in this paper is, if the L$_\infty$ relations
 uniquely fix the products or whether  there are other choices
$\tilde{b}_2(\lambda,A)$ and $\tilde{b}_2(A,B)$ solving them. Note that
this is in general a tough question and will very much depend on the
specific theory we are looking at. Hence we do not expect an answer to
this question solely from the L$_\infty$ perspective. In this paper we
rather want to address the situation when there are different
solutions to the equations, and discuss how to decide if they lead to
equivalent field theories. 

From what we already said, we expect
that two solutions related by a SW map should be considered physically
equivalent. It would  be nice if mathematically this corresponded to
the existence of a corresponding L$_\infty$ QISO. Before we formulate 
a general theorem relating the physical and the mathematical approach, let
us discuss a simple example.

\subsection{Redundancies in the abelian Chern-Simons theory}
\label{sec_example}

Let us consider  the L$_\infty$ algebra for a free 3D Chern-Simons
theory \cite{Hohm:2017pnh} with gauge group $U(1)$. The vector space
is $X=X_{1}\oplus X_0\oplus X_{-1}$ and  the only non-trivial maps are
\eq{
b_1(\lambda)_a &= \partial_a \lambda\\
b_1(A)_a &= {\epsilon_a}^{bc}\partial_b A_c \\
b_2(\lambda_1,\lambda_2) &= 0 \,.
}
Choosing all other maps to vanish satisfies the L$_\infty$ algebra so
the bootstrap approach is rather trivial.
The gauge transformations, gauge algebra and field equations are very simple
\eq{
\delta_\lambda A = \partial_a \lambda \;,\quad
[\delta_{\lambda_1},\delta_{\lambda_2}]=0 \;,\quad
\mathcal{F}={\epsilon_a}^{bc}\partial_b A_c \,.
}

Revisiting the defining equation (\ref{secondequ}) for inputs
$\lambda_1,\lambda_2\in X_1$, the condition on $b_2(\lambda,A)$ is not to vanish identically but rather
\eq{
0=b_2\left(b_1(\lambda_1),\lambda_2\right)- b_2\left(\lambda_1,b_1(\lambda_2)\right) \,.
}
This admits other solutions, for instance 
\eq{
\tilde{b}_2(\hat\lambda,\hat A)_a=-v^i(\hat A_i\partial_a \hat\lambda +
\hat A_a\partial_i \hat\lambda)
}
with some constant vector $v^i$. Comparing with (\ref{gaugetrafodefi})
this map plays a role in the gauge transformation of the fields $\hat A\in
X_0$. The L$_\infty$ structure forces us to introduce further maps to
satisfy all defining equations, essentially bootstrapping a theory
with deformed gauge transformations. Up to  order $\tilde{b}_3$, the maps are given by
\eq{ \label{u1cs-bmaps}
 b_1(\hat\lambda)_a &= \partial_a \hat\lambda \\
 b_1(\hat A)_a &= {\epsilon_a}^{bc}\partial_b \hat A_c \\
 \tilde{b}_2(\hat A,\hat\lambda)_a &= -v^i(\hat A_i\partial_a \hat\lambda +
 \hat A_a\partial_i \hat\lambda)\\
 \tilde{b}_2(\hat E,\hat\lambda)_a &= v^i \hat E_a \partial_i \hat\lambda \\
 \tilde{b}_2(\hat A,\hat B)_a &= v^i {\epsilon_a}^{bc}(\hat
 A_c\partial_b \hat B_i + \partial_b \hat A_i \hat B_c) \\
 \tilde{b}_3(\hat A,\hat B,\hat\lambda)_a &= -v^i v^j \hat A_i \hat B_j \partial_a \hat\lambda \\
 \tilde{b}_3(\hat A,\hat B,\hat C)_a &= 2 v^i v^j
 {\epsilon_a}^{bc}\big(\partial_b \hat A_i (\hat B_j \hat C_c+\hat B_c
 \hat C_j) +(\hat A\hat B\hat C\text{-cyclic})\big)\\
 &\vdots
}
The deformed gauge transformation can be read-off as 
\eq{
\hat{\delta}_{\hat\lambda} \hat A = \partial_a \hat\lambda - v^i \hat A_i\partial_a
\hat\lambda -v^i \hat A_a\partial_i \hat\lambda -\frac{1}{2}v^i v^j \hat
A_i \hat A_j \partial_a \lambda + ...
}
and the deformed field equation are 
\eq{
\hat{\mathcal{F}}={\epsilon_a}^{bc}\left(\partial_b  \hat
  A_c+v^i\partial_b \hat A_i \hat A_c + 2 v^iv^j \partial_b \hat A_i
  \hat A_j \hat A_c +... \right)
}
while the closure remains trivial. As the ellipsis indicate the
equations are only exact up to terms coming from maps higher than
those given above. In the following we will suppress the ellipsis for
better readability. Now let us analyze whether these two bootstrapped
solutions are related via a SW map and an L$_\infty$ QISO, respectively.

\subsubsection*{Seiberg-Witten map}

Inspection reveals that indeed there exists a Seiberg-Witten map
relating the two solutions. The  field redefinitions are 
\eq{
\hat{A}_a(A) &= A_a - v^iA_iA_a + \frac{1}{2}v^iv^jA_iA_jA_a
\\
\hat{\lambda}(\lambda,A)&= \lambda 
\;.
}
The Seiberg-Witten condition (\ref{SWMAP1}) can be easily verified by
direct computation. Since the gauge algebra is trivial in both cases
and the gauge parameters map directly to each other, the deformed
gauge algebra maps to the original gauge algebra. One can also show
that the deformed field equation map to a product of original one and
a function of the gauge field,
\eq{
\label{SWfielde}
\hat{\mathcal{F}}(\mathcal{F},A)=\mathcal{F}(A)\left(1-v^iA_i +\frac{1}{2}v^iv^jA_iA_j  \right) \,.
}

\subsubsection*{L$_\infty$ QISO}

Continuing the discussion from the end of section \ref{sec_basicL},
let us now investigate whether we can also identify an L$_\infty$ QISO.
First recall that we have 
$\tilde{b}_1=b_1$ and $b_{n\geq 2}=0$. Then the first
quasi-isomorphism relation \eqref{QISO1}  is the fact that the free theories are the
same, $F_1(b_1(x))=\tilde{b}_1(F_1(x))=b_1(x)$ with $F_1={\rm id}$ and
for $x\in\{\lambda, A\}$.
The quasi-isomorphism relation \eqref{QISO2} for $\tilde{b}_2$ is
\eq{
 \tilde{b}_2(x_1,x_2)=F_2(b_1(x_1),x_2)+(-1)^{x_1}F_2(x_1,b_1(x_2))-b_1(F_2(x_1,x_2))\,.
}
Evaluating this on the possible combinations of entries $(\lambda_1,\lambda_2)$, $(\lambda,A)$, $(\lambda,E)$, and $(A,B)$ allows us to read off the next quasi-isomorphism maps.
\eq{\label{1}
F_2(E,A)&=-v^iE_a A_i && \in X_{-1}\\
F_2(A,B)&=-v^i(A_iB_a + A_a B_i) && \in X_0 \\
F_2(E,\lambda)&=F_2(A,\lambda)=0\,.
}
The quasi-isomorphism relation for $\tilde{b}_3$ simplifies
drastically thanks to the simple form of the original L$_\infty$
algebra. Separated into knowns on the left hand side  and unknowns on the right, it reads 
\eq{
 &\tilde{b}_3(x_1,x_2,x_3)+\tilde{b}_2(x_1,F_2(x_2,x_3))\\
 &+\tilde{b}_2(F_2(x_1,x_2),x_3) +(-1)^{x_2x_3}\tilde{b}_2(F_2(x_1,x_3),x_2)
 \\[3pt]
 &= F_3(b_1(x_1),x_2,x_3)+(-1)^{x_1}F_3(x_1,b_1(x_2),x_3)\\
 &\quad+(-1)^{x_1+x_2}F_3(x_1,x_2,b_1(x_3))-b_1(F_3(x_1,x_2,x_3))\,.
}
Evaluating these on the list of inputs $(E_1,E_2,\lambda)$, $(E,A,B)$,
$(E,A,\lambda)$, $(A,B,C)$, $(A,B,\lambda)$, $(E,\lambda_1,\lambda_2)$, and $(A,\lambda_1,\lambda_2)$ we can read
off the non-vanishing QISO maps $F_3$:
\eq{
\label{2}
 F_3(E,A,B)&=v^iv^jE_a A_iB_j && \in X_{-1}\\
 F_3(A,B,C)&=v^iv^j(A_iB_jC_a + A_jB_a C_i+A_a B_iC_j) && \in X_0\,. 
}
Let us take special note of the quasi-isomorphism acting on identical gauge fields
\eq{ 
F_2(A,A)_a&=-2v^iA_iA_a \\
F_3(A,A,A)_a&=3v^iv^jA_iA_jA_a\,.
}
Noting also that except for $F_1={\rm id}$ the quasi-isomorphism acts
on gauge parameters trivially,  we see a nice connection to the Seiberg-Witten map:
\eq{
\hat{A}_a(A) &= F_1(A) + \frac{1}{2}F_2(A,A)+ \frac{1}{6} F_3(A,A,A)+\ldots
\\
\hat{\lambda}(\lambda,A)&= F_1(\lambda)
\;.
}
Moreover, the new field equation \eqref{SWfielde} can be expressed in
terms of the QISO as
\eq{
\hat{\mathcal{F}}(\mathcal{F},A)=
F_1({\cal F}) + F_2({\cal F},A)+ \frac{1}{2} F_3({\cal F},A,A)+\ldots\,.
}
In addition, the form of the maps \eqref{1} and  \eqref{2} suggests the general solution
\eq{
\label{ansatz}
F_n(A^1,\dots, A^n)_a&=(-1)^{n-1} v^{i_1}\dots
v^{i_{n-1}}\left(A^1_{i_1}\dots A^{n-1}_{i_{n-1}}\, A^n_a+ \, {\rm cycl} \, \right)\\
F_{n+1}(E,A^1,\dots, A^n)_a&=(-1)^n\, v^{i_1}\dots v^{i_n}\, E_a\, A^1_{i_1}\dots A^n_{i_n}\\[0.1cm]
F_1(\lambda)&=\lambda
}
with  all other maps being zero. By inspection of equation
\eqref{PRNEWQISO} one can see that starting with a given L$_\infty$
algebra $(V,\left\lbrace b_i \right\rbrace)$ an arbitrary set of
graded symmetric, multilinear  maps $F_n:V\rightarrow V,\, n\ge 2$ of
degree 0  together with $F_1={\rm id}$ defines a new, quasi-isomorphic
L$_\infty$ algebra $(V,\big\lbrace \hat{b}_i
\big\rbrace)$. Therefore we can take the maps \eqref{ansatz} to
complete the solution \eqref{u1cs-bmaps} to all orders. The
constructed theory is then related to the usual Chern-Simons theory
via the 
Seiberg-Witten maps
\eq{
\hat{A}_a(A)&=A_a\, \mathrm{exp}(-v^i A_i)\\
\hat{\lambda}(\lambda,A)&=\lambda\\
\hat{\mathcal{F}}(\mathcal{F},A)&=\mathcal{F}(A)\, \mathrm{exp}(-v^i A_i)\; .
}
Thus, analogously to the fact that the gauge transformations
and the field equations are encoded in the higher products $b_n$ of an 
L$_\infty$ algebra, also the SW-like field redefinitions seem to be encoded
in the $F_n$ maps of an L$_\infty$ QISO.
In the following section we will make this connection more precise.

\section{SW maps are L$_\infty$ QISOs}
\label{sec_main}

Based on the example in the last section we analyze the general
relation between Seiberg-Witten maps and L$_\infty$
quasi-isomorphisms.  Assume we are given a quasi-isomorphism
$\left\lbrace F_n \right\rbrace :V\rightarrow W$  between two gauge
theories $(V=V_1\oplus V_0\oplus V_{-1}, \left\lbrace
  b_i\right\rbrace)$, $(W=W_1\oplus W_0\oplus W_{-1}, \big\lbrace
\tilde{b}_i\big\rbrace) $. Recall that the maps of a quasi-isomorphism
are of degree $0$. Thus an element $\hat{\lambda}\in W_1$ can be of
the schematic form 
\eq{
\hat{\lambda}\sim \underset{n}{\sum}\,
F_{n+1}(\lambda,A^n)&+\underset{k}{\sum} \,
F_{k+3}(\lambda,\mu,E,A^k)\\
&+\underset{k}{\sum} \, F_{k+5}(\lambda,\mu_1,\mu_2,E_1,E_2,A^k)+\dots \; .
}
The first sum is what is expected from the point of view of a
Seiberg-Witten map. All other terms contain at least two gauge
parameters $\lambda$, $\mu$. As we are working on the level of infinitesimal gauge
transformations those terms are suppressed and we don't include them in the expansion. Similarly,
for a gauge field $\hat{A}\in W_0$ there is the expansion
\eq{
\hat{A}\sim\! \underset{n}{\sum} \, F_n(A^n)+\underset{k}{\sum}\, F_{k+2}(\mu,E,A^k)+\underset{k}{\sum} \, F_{k+4}(\mu_1,\mu_2,E_1,E_2,A^n)+\dots.
}
Again, the first term is expected. Terms in the third sum and all
higher terms are at least of second order in the gauge parameters and
therefore suppressed. From a physics point of view  the second sum does not make much sense either, as the field of the hatted theory would depend on the gauge parameter $\mu$. It is therefore reasonable to ignore those, too and take only the first sum. Next, an element $\hat{E}\in W_{-1}$ can have the expansion
\eq{
\hat{E}\sim \underset{n}{\sum}\, F_{n+1}(E, A^n)+\underset{k}{\sum} \, F_{k+3}(\mu,E_1,E,A^k)+\dots \; .
}
Elements of degree $-1$ in the L$_\infty$ algebra are related to the equations of motion. A possibly non vanishing term in the second sum would mean that the field equations of the hatted theory depend on a random gauge parameter $\mu$. Thus it is sensible to drop those.  

Note that these general considerations are consistent  with the example in the
last section where the  only non-trivial quasi-isomorphism maps were
$F_2(A,B)$, $F_2(A,E)$, $F_3(A,B,C)$ and $F_3(A,B,E)$. In addition, the discussion in the
example went  beyond that of a mere Seiberg-Witten map since we
included the field equations.  After these considerations we are now
ready to formulate a theorem that clearly relates a SW map  and an 
L$_\infty$ QISO.  The following theorem is the main result of this paper:

\begin{theo} \label{theo}
Let $(V=V_1\oplus V_0\oplus V_{-1}, \left\lbrace b_i\right\rbrace)$, $(W=W_1\oplus W_0\oplus W_{-1}, \big\lbrace  \tilde{b}_i\big\rbrace) $ be two L$_\infty$ algebras underlying two classical gauge theories. Then 
\begin{itemize}
\item[A)] There exists a Seiberg-Witten map
  $\hat{\lambda}=\hat{\lambda}(\lambda,A)$ and $\hat{A}=\hat{A}(A)$,
satisfying
\eq{\
\hat{A}(A+\delta_\lambda A)=
\hat{A}(A)+\hat{\delta}_{\hat{\lambda}(\lambda, A)}\hat{A}(A)\; 
}
and the closure mapping
\eq{
\hat{A}&(A+\delta_{C(\lambda_1,\lambda_2,A)}\, A+\delta^{T}_{C(\lambda_1,\lambda_2,\mathcal{F})} A)\\
&=\hat{A}(A)+\hat{\delta}_{\hat{C}(\hat{\lambda}_1,\hat{\lambda}_2,\hat{A})+\hat{\lambda}(\lambda_2,\delta_{\lambda_1}
A)-\hat{\lambda}(\lambda_1,\delta_{\lambda_2} A) } \, \hat{A}(A)
+\hat{\delta}^T_{\hat{C}(\hat{\lambda}_1,\hat{\lambda}_2,\mathcal{F})} \hat{A}(A)\,, 
}
between the gauge theories if and only if there exist graded symmetric maps  $\left\lbrace F_n \right\rbrace :V\rightarrow W$ of degree 0 with 
\eq{\label{theocond1}
F_{n+k+l}\left(\lambda_1,\dots,\lambda_k,E_1,\dots,E_l,A^n\right)=0, \quad {\rm for\ all\ } \; k,l\in \left\lbrace 1,2 \right\rbrace, \; n\ge 0
}
satisfying the L$_\infty$ quasi-isomorphism relations for inputs $(\lambda, A^n)$,$(\lambda_1,\lambda_2,A^n)$,\linebreak
$(\lambda_1,\lambda_2,E,A^n)$ for all $n\ge 0$.

\item[B)] The Seiberg -Witten map of A) maps the dynamics of the field theories according to
\eq{
\hat{\mathcal{F}}&=\hat{\mathcal{F}}(\mathcal{F},A), \\
\hat{\mathcal{F}}(\mathcal{F}+\delta_\lambda \mathcal{F}, A+\delta_\lambda A)&=\hat{\mathcal{F}}(\mathcal{F},A)+\hat{\delta}_{\hat{\lambda}(\lambda,A)} \hat{\mathcal{F}}(\mathcal{F},A)
} 
if and only if the graded symmetric maps of A) satisfy the L$_\infty$ quasi-isomorphism relations on inputs $(A^n)$,$(\lambda,E, A^n)$ for all $n\ge 0$.
\end{itemize}
\end{theo} 

\vspace{0.2cm}
\noindent
In the remainder of this section we prove this theorem.
\begin{proof}
We start with the proof of $A)$. Assume that there exist graded symmetric maps $\left\lbrace F_n\right\rbrace:V\rightarrow W$ of degree $0$ s.th. \eqref{theocond1} holds and the maps satisfy the L$_\infty$ quasi-isomorphism relations on the inputs stated in $A)$. We then define
the corresponding Seiberg-Witten maps as
\eq{
\label{DEFISW}
\hat{A}(A)=\underset{n=1}{\overset{\infty}{\sum}}\, \frac{1}{n!}\, F_n(A^n)\,,\qquad
\hat{\lambda}(\lambda,A)=\underset{k=0}{\overset{\infty}{\sum}}\, \frac{1}{k!}\, F_{k+1}(\lambda,A^k)\;.
}
Recall from section 3 the form of the gauge variations \eqref{gaugetrafodefi} in terms of L$_\infty$ brackets
\eq{
\label{VARDEF}
\delta_\lambda\, A=\underset{n=0}{\overset{\infty}{\sum}}\, \frac{1}{n!}\, b_{n+1}(\lambda,A^n)\,,\qquad
\hat{\delta}_{\hat{\lambda}}\, \hat{A}=\underset{n=0}{\overset{\infty}{\sum}}\, \frac{1}{n!}\,\tilde{b}_n(\hat{\lambda},\hat{A}^n)\;.
}
Using the defining relations of an L$_\infty$ quasi-isomorphism we
prove a first lemma that says that gauge orbits are mapped to gauge orbits.
\begin{lem}
\eq{\label{SWMAP}
\hat{A}(A+\delta_\lambda \, A)=\hat{A}(A)+\hat{\delta}_{\hat{\lambda}}\, \hat{A}(A)
}
\end{lem}
\begin{proof}
We start with the left hand side of \eqref{SWMAP} for which we can
calculate
\begin{eqnarray}
\label{SWLHS}
\hat{A}(A+\delta_\lambda A)&=&\hat{A}\left(A+\underset{k=0}{\overset{\infty}{\sum}}\, \frac{1}{k!}\, b_{k+1}(\lambda,A^k)\right)\nonumber\\
&=&\underset{n=1}{\overset{\infty}{\sum}}\, \frac{1}{n!}\, F_n\bigg(A+\underset{k=0}{\overset{\infty}{\sum}}\, \frac{1}{k!}\, b_{k+1}(\lambda,A^k),\dots,A+\underset{k=0}{\overset{\infty}{\sum}}\, \frac{1}{k!}\,b_{k+1}(\lambda,A^k)\bigg)\nonumber\\
&=&\underset{n=1}{\overset{\infty}{\sum}}\, \frac{1}{n!}\,\left[ F_n\left(A^{n}\right)+n\, F_n\left(A^{n-1},\underset{k=0}{\overset{\infty}{\sum}}\, \frac{1}{k!}\,b_{k+1}(\lambda,A^k)\right)+\mathcal{O}(\lambda^2)\right]\nonumber\\
&=&\hat{A}(A)+\underset{n=0}{\overset{\infty}{\sum}}\, \frac{1}{n!}\,\left[ F_{n+1}\left(\underset{k=0}{\overset{\infty}{\sum}}\, \frac{1}{k!}\,b_{k+1}(\lambda,A^k),A^n\right)\right]\\
&=&\hat{A}(A)+\underset{m=0}{\overset{\infty}{\sum}}\,\frac{1}{m!}\underset{k+n=m}{\sum}\, \frac{m!}{n!k!}\, F_{n+1}\left(b_{k+1}(\lambda,A^k),A^n\right)\nonumber\,.
\end{eqnarray}
The first two equalities are just writing the variation in terms of
L$_\infty$ brackets and using the definition of $\hat{A}(A)$ in terms
of the $F_n$'s. In the next equality the linearity of the map $F_n$ is
used and we truncated to linear order in the infinitesimal gauge
parameter $\lambda$. Next we changed the index of the outer sum and switched
the order of the entries, which doesn't cause extra minus signs as all
elements are of degree $0$. In the last step we just rewrote the sums
in a more convenient form and inserted $\frac{m!}{m!}$. 

To proceed, we recall the left hand side 
of the defining equation for an L$_\infty$ quasi-isomorphism \eqref{defiquasi}:
\eq{
\!\!\!\!\!\!\underset{\sigma \in\, {\rm Unsh}(k+n=m+1)}{\sum} \, \epsilon(\sigma; x)\, F_{1+n}\Big(b_{k+1}(x_{\sigma(1)},\dots, x_{\sigma(k+1)}),x_{\sigma(k+2)},\dots, x_{\sigma(m+1)}\Big)\,.
}
Choosing the input $x_1=\lambda,\, x_2=\dots=x_{m+1}=A$ and taking
into account that interchanging any two element never causes a minus
sign this expression becomes
\eq{\label{QUSIAF}
\underset{k+n=m}{\sum}\, \frac{m!}{k!n!}\,
F_{1+n}&\Big(b_{k+1}(\lambda,A^k),A^n\Big)\\
&+\frac{m!}{(k+1)!(n-1)!}\, F_{1+n}\Big(b_{k+1}(A^{k+1}),\lambda,A^{n-1}\Big)\;.
}
In the first summand the prefactor $\frac{m!}{k!n!}$ is the number of
unshuffles for the $m$ gauge fields $A$ into partitions of length $k$ and $n$. The
same holds for the second summand\footnote{This is slightly cheating,
  since there is a term with $n=0$ and we set (-1)!=1. }. Upon
bringing the second term in \eqref{QUSIAF} to the right hand side, the
left hand side  of the
defining relation for an L$_\infty$ quasi-isomorphism appears in
\eqref{SWLHS}. We first investigate what happens to the second term in \eqref{QUSIAF}:
\eq{
&\,\underset{m=0}{\overset{\infty}{\sum}}\frac{1}{m!}\; \underset{k+n=m}{\sum}\, \frac{m!}{(k+1)!(n-1)!}\, F_{n+1}\left(b_{k+1}(A^{k+1},\lambda,A^{n-1}\right)\\
=&\,\underset{m=0}{\overset{\infty}{\sum}}\; \underset{k+n=m+1}{\sum}\,\frac{1}{k!(n-1)!}\, F_{n+1}\left(b_k(A^k),\lambda,A^{n-1}\right)\\
=&\,\underset{n=0}{\overset{\infty}{\sum}}\,\frac{1}{(n-1)!} F_{n+1}\left(\underset{k=1}{\overset{\infty}{\sum}}\,\frac{1}{k!}\, b_k(A^k),\lambda,A^{n-1}\right)\\
=&\,\underset{n=0}{\overset{\infty}{\sum}}\,\frac{1}{(n-1)!} F_{n+1}\left(\mathcal{F},\lambda,A^{n-1}\right)=0\;.
}
In the last step the equation of motion is abbreviated by
$\mathcal{F}$. Thus, we realize that  in general there will appear more terms on the
right hand side  of the SW-condition \eqref{SWMAP}. However,  these
terms are proportional to 
the equation of motion and are of the type appearing in 
\eqref{theocond1} so that they actually vanish.

Employing now the quasi-isomorphism equation \eqref{SWLHS} for
 $x_1=\lambda, x_2=\dots=x_{n+1}=A$, one can express \eqref{SWLHS}
as
\eq{\label{ZWISCHEN}
\hat{A}(A)+\underset{m=0}{\overset{\infty}{\sum}}\, \frac{1}{m!}\,
\underset{\sigma\in {\rm Unsh}(k_1+\dots+k_j=m+1)}{\sum}\, \frac{\epsilon(x;\sigma)}{j!}\, \tilde{b}_j\left(F_{k_1}\otimes\dots\otimes F_{k_j}\right)(x_{\sigma(K)})\;.
}
We still stated the equation in an abstract form to highlight the
point where \eqref{defiquasi} is used. 
In order to unwrap the expression we note that 
\eq{
\Big|{\rm Unsh}(k_1+\dots+k_j=n)\Big|=\binom{n}{k_1,\dots,k_j}
}
meaning that there are $\binom{n}{k_1,\dots,k_j}$ possibilities to
order a set of size $n$ into $j$ partitions of length $k_i,\,
i=\left\lbrace1,\dots, j\right\rbrace$ preserving the order in each
partition. 
Using this we get for the second term in \eqref{ZWISCHEN}
\begin{eqnarray}
&&\!\!\!\!\underset{m=0}{\overset{\infty}{\sum}}\,
\frac{1}{m!}\,\underset{k_1+\dots+k_j=m+1}{\sum}\,
\frac{1}{j!}\left[\tilde{b}_j\big((F_{k_1}(\lambda,A^{k_1-1}),F_{k_2}(A^{k_2}),\dots,
  F_{k_j}(A^{k_j})\big)
{\textstyle \binom{m}{(k_1-1),k_2,\dots,k_j}}\right. \nonumber\\[0.1cm]
&&\hspace{2cm}+\tilde{b}_j\big(F_{k_1}(A^{k_1}),F_{k_2}(\lambda,A^{k_2-1}),\dots,F_{k_j}(A^{k_j}\big)
{\textstyle \binom{m}{k_1,(k_2-1),\dots,k_j}}\nonumber\\
&&\hspace{2cm}\vdots \\
&&\hspace{2cm}+\left. \tilde{b}_j\big(F_{k_1}(A^{k_1}),\dots,
  F_{k_j}(\lambda,A^{k_j})\big) 
{\textstyle \binom{m}{k_1,\dots, (k_j-1)}}\right]\;.\nonumber
\end{eqnarray}
Of course the terms in the squared bracket are all the same. Using the
definition of the multinominal-coefficient this gives 
\eq{\label{ZWISCHEN2}
&\underset{m=0}{\overset{\infty}{\sum}}\,\underset{k_1+\dots+k_j=m+1}{\sum}\,
{\textstyle \frac{1}{(j-1)!}\,\frac{1}{(k_1-1)!\cdots k_j!}}
\, \tilde{b}_j\big(F_{k_1}(\lambda,A^{k_1-1}),F_{k_2}(A^{k_2}),\dots, F_{k_j}(A^{k_j})\big)\\
=&\underset{m=0}{\overset{\infty}{\sum}}\,\underset{k_0+\dots+k_{j}=m}{\sum}\,\frac{1}{j!}\,\frac{1}{k_0!\cdots k_{j}!}\, \tilde{b}_{j+1}\big(F_{k_0+1}(\lambda,A^{k_1}),F_{k_1}(A^{k_1}),\dots, F_{k_{j}}(A^{k_{j}})\big)\\
=&\underset{j=0}{\overset{\infty}{\sum}}\frac{1}{j!}\, \underset{k_0\ge 0}{\sum}\, \underset{k_1,\dots,k_{j}\ge 1}{\sum}\, \frac{1}{k_0!\cdots k_{j}!}\, \tilde{b}_{j+1}\big((F_{k_0+1}(\lambda,A^{k_0}),F_{k_1}(A^{k_1}),\dots, F_{k_{j+1}}(A^{k_{j}})\big)
}
where we changed the summation indices $j$ and $k_1$ in the first step
and rewrote the summation in a more convenient way in the second. 

All this was for the left  hand side of \eqref{SWMAP}.
Next we compute the second term on the right hand  of \eqref{SWMAP}
\begin{eqnarray}
&&\hat{\delta}_{\hat{\lambda}}\, \hat{A}(A)=\underset{j=0}{\overset{\infty}{\sum}}\, \frac{1}{j!}\, \tilde{b}_{j+1}\Big(\hat{\lambda}(\lambda,A),\hat{A}(A)^j\Big)\nonumber\\
&&\hspace{1.5cm}=\underset{j=0}{\overset{\infty}{\sum}}\, \frac{1}{j!}\, \tilde{b}_{j+1}\left(\underset{n=0}{\sum^\infty}\frac{1}{n!}\, F_{n+1}(\lambda,A^n),\Big(\underset{k=1}{\sum^\infty}\frac{1}{k!}\, F_k(A^k)\Big)^j\right)\\
&&\hspace{0.3cm}=\underset{j=0}{\overset{\infty}{\sum}}\, \frac{1}{j!} \,
\underset{n=0}{\sum^\infty}\, \underset{k_1,\dots, k_j \ge 1}{\sum}\,
 \frac{1}{n!}\frac{1}{k_1!\cdots k_j!}\, \tilde{b}_{j+1}\big(F_{n+1}(\lambda,A^n),F_{k_1}(A^{k_1}),\dots, F_{k_j}(A^{k_j})\big)\;.\nonumber
\end{eqnarray}
But this is exactly \eqref{ZWISCHEN2} which is equivalent to the
second term in \eqref{ZWISCHEN}. This proves Lemma 1.
\end{proof}

Next we want to derive the closure statement. Note that we can be  more general than in section 3 and consider a gauge algebra which closes only on shell, i.e. 
\eq{
\left[\delta_{\lambda_1}, \delta_{\lambda_2}\right] A=\delta_{C(\lambda_1,\lambda_2,A)}\, A+\delta^{T}_{C(\lambda_1,\lambda_2,\mathcal{F})} A\; . 
}
If we ignore this term for the moment we can prove the guess of section 3.
\begin{lem}
\eq{\label{dashalt}
\hat{A}(A+\delta_{C(\lambda_1,\lambda_2,A)}\,
A)&=\hat{A}(A)+\hat{\delta}_{\hat{C}(\hat{\lambda}_1,\hat{\lambda}_2,\hat{A})+\hat{\lambda}(\lambda_2,\delta_{\lambda_1}
A)-\hat{\lambda}(\lambda_1,\delta_{\lambda_2} A) } \, \hat{A}(A)
}
\end{lem}
\begin{proof}
Upon using lemma 1 this amounts to
\eq{
\hat{\lambda}(C(\lambda_1,\lambda_2,A),A)=\hat{C}(\hat{\lambda}_1,\hat{\lambda}_2,\hat{A})+\hat{\lambda}(\lambda_2,\delta_{\lambda_1} A)-\hat{\lambda}(\lambda_1,\delta_{\lambda_2} A)\; .
}
This equality is readily proven by going through the analog steps as in the proof of Lemma 1.
\end{proof}

Now we check what happens when  we include the term
$\delta^{T}_{C(\lambda_1,\lambda_2,\mathcal{F})} A$  in
the closure condition. Since the gauge algebras in the original
Seiberg-Witten map closed off-shell, there is no obvious guess. We
will utilize the quasi-isomorphism to derive a transformation rule. We start with 
 \eq{
 \hat{A}\big(A+\delta_{C(\lambda_1,\lambda_2,A)}\,
 A&+\delta^{T}_{C(\lambda_1,\lambda_2,\mathcal{F})} A\big)\\
&=\underset{n=1}{\overset{\infty}{\sum}}\, \frac{1}{n!}\, F_n\big((A+\delta_{C(\lambda_1,\lambda_2,A)}\, A+\delta^{T}_{C(\lambda_1,\lambda_2,\mathcal{F})} A)^n \big)\\
 &=\underset{n=1}{\overset{\infty}{\sum}}\, \bigg[\frac{1}{n!}\, F_n(A^n)+\frac{1}{(n-1)!}\, F_n\big(\delta_{C(\lambda_1,\lambda_2,A)}\, A, A^{n-1}\big)\\
 &\phantom{===}+\frac{1}{(n-1)!}\, F_n\big(\delta^{T}_{C(\lambda_1,\lambda_2,\mathcal{F})} A, A^{n-1}\big)+\mathcal{O}(\lambda^3)\bigg]\; ,
 }
where  the first two terms are the ones we already computed for \eqref{dashalt}.  The third term is new. Inserting \eqref{defieomtrafo}, using the defining equation for an L$_\infty$ morphism and conditions \eqref{theocond1} this term is computed to be 
\eq{
\underset{n=1}{\overset{\infty}{\sum}}\frac{1}{(n-1)!}\,
&F_n\big(\,\delta^{T}_{C(\lambda_1,\lambda_2,\mathcal{F})} A, A^{n-1}\big)\\[-0.1cm]
&=\underset{n=1}{\overset{\infty}{\sum}}\, \frac{1}{n!}\, \tilde{b}_n(\hat{\lambda}_1,\hat{\lambda}_2, \hat{\mathcal{F}}, \hat{A}^n)=\hat{\delta}^T_{\hat{C}(\hat{\lambda}_1,\hat{\lambda}_2, \hat{\mathcal{F}})}\, \hat{A}
}
Thus we get the transformation rule
\eq{
\label{SWclosuredef}
\hat{A}(A&+\delta_{C(\lambda_1,\lambda_2,A)}\, A+\delta^{T}_{C(\lambda_1,\lambda_2,\mathcal{F})} A)\\
&=\hat{A}(A)+\hat{\delta}_{\hat{C}(\hat{\lambda}_1,\hat{\lambda}_2,\hat{A})+\hat{\lambda}(\lambda_2,\delta_{\lambda_1}
A)-\hat{\lambda}(\lambda_1,\delta_{\lambda_2} A) } \, \hat{A}(A)
+\hat{\delta}^T_{\hat{C}(\hat{\lambda}_1,\hat{\lambda}_2,\mathcal{F})} \hat{A}(A)\; .
}
This proves one  direction of the Theorem 1A). For the other direction
note that we can start with a Seiberg-Witten map $\hat{A}(A)$  and
$\hat{\lambda}(\lambda,A)$ of the
form  \eqref{DEFISW} satisfying the relations \eqref{SWMAP1} and
\eqref{SWclosuredef}.
 Going all the steps in the computations backwards then reveals that
 the $F_n$ with the same gauge field inserted satisfy the defining equation of an L$_\infty$
 quasi-isomorphism on the stated inputs with the trivial assertions  \eqref{theocond1}.
Using the graded symmetry and polarization identities (see \cite{Hohm:2017pnh}) this is enough to get the equations on general inputs
of gauge fields. This finally proves Theorem 1A).

For the proof of Theorem 1B) we have to perform very similar computations. 
Recall that the equation of motion is expanded as
\eq{\label{EOM}
\mathcal{F}=\underset{n=1}{\overset{\infty}{\sum}}\,\frac{1}{n!}\, b_n\left(A^n\right)\; .
}
Given the maps $\left\lbrace F_n \right\rbrace:V\rightarrow W$ of $A)$, which we assume to satisfy additionally the L$_\infty$ QISO relations on inputs $(A^n)$, $(\lambda,E,A^n)$, we define for the equation of motion in the hatted theory
\eq{\label{hattedeom}
\hat{\mathcal{F}}=\underset{n=0}{\overset{\infty}{\sum}}\, \frac{1}{n!}\, F_{n+1}\left(\mathcal{F},A^n\right)\; .
}
In order for this to be consistent we have to show
\begin{lem}
\eq{\label{TOSHOWEOM}
\underset{n=0}{\overset{\infty}{\sum}}\, \frac{1}{n!}\, F_{n+1}\left(\mathcal{F},A^n\right)=\underset{n=1}{\overset{\infty}{\sum}}\,\frac{1}{n!}\, \tilde{b}_n\left(\hat{A}^n\right)
}
with $\hat{A}(A)$ given by \eqref{DEFISW}. 
\end{lem}
\begin{proof}
Inserting \eqref{EOM} in the left hand side  of \eqref{TOSHOWEOM} and using the defining relation for an L$_\infty$ morphism this is can be verified by a similar computation as in the proof of Lemma 1.
\end{proof}

\begin{lem}
\eq{\label{EOMtrafo}
\hat{\mathcal{F}}(\mathcal{F}+\delta_\lambda\mathcal{F}, A+\delta_\lambda A)&=\hat{\mathcal{F}}(\mathcal{F},A)+\hat{\delta}_{\hat{\lambda}(\lambda,A)} \hat{\mathcal{F}}(\mathcal{F},A)\; .
}
\end{lem}
\begin{proof}
We start by computing the left hand side of \eqref{EOMtrafo}.
\eq{
&\hat{\mathcal{F}}(\mathcal{F}+\delta_\lambda \mathcal{F}, A+\delta_\lambda A)\\
=&\underset{n=0}{\sum^\infty} \, \frac{1}{n!}\bigg[ F_{n+1}\left(\mathcal{F},A^n\right)+F_{n+2}\left(\mathcal{F},\delta_\lambda A,A^n\right)+F_{n+1}\left(\delta_\lambda \mathcal{F}, A^n\right)\bigg]+\mathcal{O}(\lambda^2)\\
=&\hat{\mathcal{F}}(\mathcal{F},A)+\underset{m=0}{\sum^\infty}\, \frac{1}{m!}\underset{n+k=m}{\sum}\,\frac{m!}{n!k!}\bigg[F_{n+2}\left(b_{k+1}(\lambda,A^k),\mathcal{F},A^n\right)\\
&\phantom{====================} + F_{n+1}\left(b_{k+2}(\lambda,\mathcal{F},A^k),A^n\right)\bigg]\,.
}
In the second equality we used \eqref{hattedeom} and inserted the
definition of the gauge transformations for the field and the
equations of motion. Next we use \eqref{defiquasi} in the second
summand which  upon using \eqref{theocond1} and some combinatorics  is equivalent to
\begin{eqnarray}
&&\underset{m=0}{\sum^\infty} \, \frac{1}{m!}\, \underset{j=0}{\sum^\infty}\, \underset{k_1+\dots +k_{j+2}=m}{\sum}\, \frac{1}{j!k_1!\dots k_{j+2}!}\, \tilde{b}_{j+2}\bigg(F_{k_1+1}(\lambda,A^{k_1}),F_{k_2+1}(\mathcal{F},A^{k_2}),\nonumber\\
&&\phantom{=======================} F_{k_3}(A^{k_3}),\dots, F_{k_{j+2}}(A^{k_{j+2}})\bigg)\nonumber\\
&&= \underset{j=0}{\sum^\infty}\, \frac{1}{j!}\, \tilde{b}_{j+2}\left(\hat{\lambda}(\lambda,A),\hat{\mathcal{F}}(\mathcal{F},A),\hat{A}(A)^j\right)\\[0.1cm]
&&=\, \hat{\delta}_{\hat{\lambda}(\lambda,A)}\, \hat{\mathcal{F}}(\mathcal{F},A)\; .\nonumber
\end{eqnarray}
This finally shows \eqref{EOMtrafo}. 
\end{proof}

For the other direction of Theorem1B) we note that we can make the
ansatz \eqref{hattedeom} and just go all the steps of the above
computations backwards. This gives the defining relations for an
L$_\infty$ quasi-isomorphism on $(A^n)$ and $(\lambda, E, A^n)$.
\end{proof}

Note that it seems that we are only using the conditions for an L$_\infty$ morphism. But the Seiberg-Witten map should of course be invertible, which implies quasi-isomorphism on the L$_\infty$ side.

\section{Conclusion}

Motivated by redundancies in defining the L$_\infty$ structure in the
bootstrap approach to gauge theories, we showed how mathematical
equivalence of the solutions implies also 
physical equivalence. This is done by showing that a
quasi-isomorphism between L$_\infty$ algebras of two
gauge field theories is equivalent to the existence of a Seiberg-Witten map
between the two. This ensures that there are the same degrees of
freedom in both field theories.  Note that when considering only the gauge L$_\infty$ algebras of the theories, i.e. setting $X_{-1}=0$, the conditions \eqref{theocond1} are trivially satisfied and we get a complete quasi-isomorphism between the gauge L$_\infty$ algebras from a Seiberg-Witten map.

In addition we derived a condition for
the closure of the gauge algebra in terms of a Seiberg-Witten
map. This was motivated by the original example of Seiberg-Witten
discussed in \cite{Seiberg:1999vs}, but we think that the equivalence
in terms of L$_\infty$ algebras establishes \eqref{SWclosuredef} as the
correct formula. Furthermore we argued that the existence of a
quasi-isomorphism of the full theory implies that the equations of
motion in both theories get mapped onto each  other.

\vspace{0.8cm}

\noindent
\subsubsection*{Acknowledgments}
We are indebted to  Andreas Deser for early discussions on this subject. In addition we want to thank Jim Stasheff for correspondence on the first version of the paper. 
 V.K. acknowledges the CAPES-Humboldt Fellowship No.~0079/16-2 and CNPq Grant No.~305372/2016-5.

\newpage

\bibliographystyle{utphys}
\bibliography{references}
 
\end{document}